\documentclass[a4paper,UKenglish,cleveref, autoref, thm-restate]{lipics-v2021}

\title{FPT Approximation using Treewidth: Capacitated Vertex Cover, Target Set Selection and Vector Dominating Set}
\titlerunning{FPT Approximation using Treewidth: CVC, TSS and VDS} 

\author{Huairui Chu}{Nanjing University}{huairuichu@163.com}{}{}
\author{Bingkai Lin}{Nanjing University}{lin@nju.edu.cn}{}{}

\authorrunning{H.Chu and B.Lin} 

\Copyright{Huairui Chu and Bingkai Lin} 



\category{} 

\relatedversion{} 

\nolinenumbers 

\usepackage[ruled]{algorithm2e}

\newcommand{\rcd}{R_{\alpha}} 
\newcommand{\rcds}[1]{R_{\alpha_{#1}}} 
\newcommand{\rcdap}{\hat{R}_{\alpha}} 
\newcommand{\rcdaps}[1]{\hat{R}_{\alpha_{#1}}} 

\newcommand{\tod}{d} 

\EventEditors{Satoru Iwata and Naonori Kakimura}
\EventNoEds{2}
\EventLongTitle{34th International Symposium on Algorithms and Computation (ISAAC 2023)}
\EventShortTitle{ISAAC 2023}
\EventAcronym{ISAAC}
\EventYear{2023}
\EventDate{December 3-6, 2023}
\EventLocation{Kyoto, Japan}
\EventLogo{}
\SeriesVolume{283}
\ArticleNo{18}

\begin{document}

\maketitle

\begin{abstract}
Treewidth is a useful tool in designing graph algorithms. Although many NP-hard graph problems can be solved in linear time  when the input graphs have small treewidth, there are problems which remain hard on graphs of bounded treewidth. In this paper, we consider three vertex selection problems that are W[1]-hard when parameterized by the treewidth of the input graph, namely the capacitated vertex cover problem, the target set selection problem and the vector dominating set problem.  We provide two new methods to obtain FPT approximation algorithms for  these problems. For the capacitated vertex cover problem and the vector dominating set problem, we obtain $(1+o(1))$-approximation FPT algorithms. 
For the target set selection problem, we give an FPT algorithm providing a tradeoff between its running time and the approximation ratio.


\end{abstract}

\section{Introduction}
\label{Introduction}
We consider  problems whose goals are to select a minimum sized vertex set in the input graph that can ``cover'' all the target objects. 
In the capacitated vertex cover problem (CVC), we are given a graph $G$ with a capacity function $c : V(G)\to\mathbb{N}$, the goal is to find a  set $S\subseteq V(G)$ of minimum size such that every edge of $G$ is covered\footnote{An edge $e$ can be covered by a vertex $v$ if $v$ is an endpoint of $e$.} by some vertex in $S$ and each vertex $v\in S$ covers at most $c(v)$ edges. This problem has application in planning experiments on redesign of known drugs involving glycoproteins~\cite{guha2002capacitated}. In the target set selection problem (TSS), we are given a graph $G$ with a threshold function $t:V(G)\to\mathbb{N}$. The goal is to select a minimum sized set $S\subseteq V(G)$ of vertices that can activate all the vertices of $G$. The activation process is defined as follows. Initially, all vertices in the selected set $S$ are activated. In each round, a vertex $v$ gets active if there are $t(v)$ activated vertices in its neighbors. The  study of  TSS has application in maximizing influence in social network~\cite{kempe2003maximizing}. Vector dominating set (VDS) can be regarded as a ``one-round-spread'' version of TSS, where the input consists of a graph $G$ and a threshold function $t:V(G)\to\mathbb{N}$, and the goal is to find a set $S\subseteq V(G)$ such that for all vertices $v\in V$, there are at least $t(v)$ neighbors of $v$ in $S$.

Since CVC generalizes the vertex cover problem, while TSS and VDS  are no easier than the dominating set problem\footnote{For VDS, when $t(v)=1$ for every vertex $v$ in the graph, VDS is the dominating set problem. For TSS, a reduction from dominating set to TSS can be found in the work of Charikar et al.~\cite{CharikarNW16}.}, 
they are both NP-hard and thus have no polynomial time algorithm unless $P=NP$.  
Polynomial time approximation algorithms for capacitated vertex cover problem have been studied extensively~\cite{guha2002capacitated,chuzhoy2006covering,gandhi2006improved,saha2012set,cheung2014improved,wong2017tight,shiau2017tight}.  
The problem has a $2$-approximation polynomial time algorithm~\cite{gandhi2006improved}. Assuming the \emph{Unique Game Conjecture}, there is no polynomial time  algorithm for the vertex cover problem with  approximation ratio better than $2$~\cite{khot2008vertex}. As for the TSS problem, 
 it is known that the minimization version of TSS cannot be approximated to $2^{\log^{1-\epsilon}n}$ assuming $P\neq NP$, or $n^{0.5-\epsilon}$ assuming the conjecture on planted dense subgraph~\cite{chen2009approximability,charikar2016approximating}.
Cicalese et al. proved that VDS cannot be approximated within a factor of $c\ln n$ for some $c$ unless $P= NP$ \cite{vds13}.

Another way of dealing with hard computational problems is to use parameterized algorithms. 
For any input instance $x$ with parameter $k$, an algorithm with running time upper bounded by $f(k)\cdot |x|^{O(1)}$ for some computable function $f$ is called FPT. 
A natural parameter for a computational problem is the solution size. 
The first FPT algorithm with running time $1.2^{k^2}+n^2$ for capacitated vertex cover problem parameterized by solution size was provided in \cite{guo2005parameterized}. In~\cite{cdscvc08}, the authors  gave an improved FPT algorithm that runs in $k^{3k}\cdot |G|^{O(1)}$.  
However, using the solution size as parameter might be too strict for CVC. Note that CVC instances with sublinear capacity functions cannot have small sized solutions.
On the other hand, TSS parameterized by its solution size is W[P]-hard
\footnote{The well known W-hierarchy is $FPT\subseteq W[1]\subseteq W[2]\subseteq ...\subseteq W[P]$, where $FPT$ denotes the set of problems which admits FPT algorithms. The basic conjecture on parameterized complexity is $FPT\neq W[1]$. We refer the readers to~\cite{downey2013fundamentals,flugro06,cygan2015parameterized} for more details.}
according to \cite{tsswp}. VDS is W[2]-hard since it generalizes the dominating set problem.

In this paper, we consider these problems   parameterized by the treewidth~\cite{robertson1986graph} of the input graph.  In fact, since the treewidth of a graph having $k$-sized vertex cover is also upper-bounded by $k$~\cite{cdscvc08}, CVC parameterized by treewidth can be regarded as a natural generalization of CVC parameterized by solution size. And it is already proved in  \cite{cdscvc08}  that CVC parameterized only by the treewidth of its input graph has no FPT algorithm assuming $W[1]\neq FPT$. 
As for the TSS problem, it can be solved in $n^{O(w)}$ time for graphs with $n$ vertices and treewidth bounded by $w$ and has no $n^{o(\sqrt{w})}$-time algorithm unless ETH fails~\cite{tssw1}.
VDS is also $W[1]$-hard when parameterized by treewidth \cite{bddw1fpt}, however, it admits an FPT algorithm with respect to the combined parameter $(w+k)$\cite{vdsw1fpt}. 







%



Recently, the approach  of combining parameterized algorithms and approximation algorithms has received increased attention~\cite{feldmann2020survey}. It is natural to ask whether there exist FPT algorithms for these problems with approximation ratios better than that of the polynomial time algorithms. Lampis~\cite{lampis} proposed a general framework for approximation algorithms on tree decomposition. 
Using his framework, one can obtain algorithms for CVC and VDS which outputs a solution of size at most $opt(I)$ on input instance $I$ but may slightly violate the capacity or the threshold requirement within a factor of $(1\pm\epsilon)$. 
However,  the framework of Lampis can not  be directly used to find an approximation solution for these problems satisfying all the capacity or threshold requirement. The situation becomes worse in the TSS problem, as  the error might propagate during the activation process. We overcome these difficulties and give positive answer to the aforementioned question. For the CVC and VDS problems, we obtain $(1+o(1))$-approximation FPT algorithms respectively.

\begin{theorem}\label{main}
    There exists an algorithm \footnote{The algorithm can be modified to output a solution with size as promised. See the remark in Appendix~\ref{mainproof1}.}, which takes
    a CVC instance $I=(G,c)$ and a tree decomposition $(T,\mathcal{X})$ with width $w$ for $G$ as input and outputs an integer $\hat{k}_{\min}\in [opt(I),(1+O(1/(w^2\log n))) opt(I)]$ in  $(w\log n)^{O(w)}n^{O(1)}$ time.
\end{theorem}

\begin{theorem}\label{thm:vdsapp}
There exists an algorithm running in time $2^{O(w^5\log w \log\log n)}n^{O(1)}$ which takes as input an instance $I=(G,t)$ of VDS and a tree decomposition of $G$ with width $w$, finds a solution of size at most $(1+1/(w\log\log n)^{\Omega(1)})\cdot opt(I)$.
\end{theorem}

Notice that the running time stated in above theorems are FPT running time, because $(\log n)^{f(w)}\leq f(w)^{O(f(w))}+n^{O(1)}$.

For the TSS problem, we give an  approximation algorithm with a tradeoff between the approximation ratio and its running time.
\begin{theorem}\label{thm:tssapp}
For all $C\in \mathbb{N}$, there is an algorithm which takes as input an instance $I=(G,t)$ of TSS and a tree decomposition of $G$ with width $w$, finds a solution of size $(1+(w+1)/(C+1))\cdot opt(I)$ in  time $n^{C+O(1)}$.
\end{theorem}









\noindent\textbf{Open problems and future work.} Note that our FPT approximation algorithm for TSS 
has ratio equal to the treewidth of the input graph. An immediate question is whether this problem has parameterized $(1+o(1))$-approximation algorithm. We remark that the reduction from $k$-Clique to TSS in~\cite{tssw1} does not preserve the gap. Thus it does not rule out constant FPT approximation algorithm for TSS on bounded treewidth  graphs even under hypotheses such as \textit{parameterized inapproximablity hypothesis} ({\sf PIH}) \cite{LRSZ20} or {\sf GAP-ETH}~\cite{Din16,MR16}.

In the regime of exact algorithms, we have the famous Courcelle’s Theorem which states that all problems defined in \emph{monadic second order logic} have linear time algorithm on graphs of bounded treewidth~\cite{arnborg1991easy,courcelle1990graph}. It is interesting to ask if  one can obtain a similar algorithmic meta-theorem~\cite{Grohe08} for approximation algorithms.

\subsection{Overview of our techiniques}
\noindent\textbf{Capacitated Vertex Cover.} Our starting point is the exact algorithm for CVC on graphs with treewidth $w$ in $n^{\Theta(w)}$ time. The exact algorithm has running time $n^{\Theta(w)}$ because it has to maintain a set of $(w+1)$-dimension vectors $\tod:X_\alpha\to[n]$ for every node $\alpha $ in the tree decomposition.  One can get more insight by checking out the exact algorithm for CVC in Section \ref{Exact Algorithm}.  To  reduce the size of such a table, Lampis' approach~\cite{lampis} is to pick a parameter $\epsilon\in(0,1)$ and round every integer to the closest integer power of $(1+\epsilon)$. In other words,  an integer $n$ is represented by $(1+\epsilon)^x$  with $(1+\epsilon)^x\le n<(1+\epsilon)^{x+1}$. Thus it suffices to keep $(\log n)^{O(w)}$ records for every bag in the tree decomposition. The price of this approach is that we can only have  approximate values for records in the table. Note that the errors of approximate values might accumulate after addition (See Lemma~\ref{error}). Nevertheless, we can choose a tree decomposition with height $O(w^2\log n)$ and set $\epsilon=1/poly(w\log n)$ so that if the dynamic programming procedure only involves adding and passing values of these vectors, then we can have $(1+o(1))$-approximation values for all the records in the table.  

Unfortunately, in the forgetting node for a vertex $v$, we need to compare the value of $\tod(v)$ and the capacity value $c(v)$. This task seems impossible if we do not have the exact value of $\tod(v)$. Our idea is to modify the ``slightly-violating-capacity'' solution, based on two crucial observations.  The first is that, in a solution, for any vertex $v\in V$, the number of edges incident to $v$ which are \textbf{not} covered by $v$ presents a lower bound for the solution size. The second observation is that one can test if a ``slightly-violating-capacity''  solution can be turned into a good one  in polynomial time. These observations are formally presented in Lemma \ref{lb} and \ref{extent}. 



\medskip

\noindent\textbf{Target Set Selection  and Vector Dominating Set.} We observe that both of the TSS and VDS problems are \emph{monotone} and \emph{splittable}, where the monotone property states that  any super set of a solution is still a solution and the splittable property means that for any separator $X$ of the input graph $G$, the union of $X$ and solutions for components after removing $X$ is also a solution for the  graph $G$. We give a general approximation for vertex subset problems that are monotone and splittable.
The key of our approximation algorithm is an observation that any bag in a tree decomposition is a separator in $G$. As the problem is splittable, we can design a procedure to find a bag, and remove it, which leads to a separation of $G$ and we then deal with the component ``rooted'' by this bag. We can use this procedure repeatedly until the whole graph is done.
\subsection{Organization of the Paper}
In Section \ref{Preliminaries} the basic notations are given, and we formally define the problem we study.
In Section \ref{Exact Algorithm} we present the exact algorithm for CVC. In Section \ref{Approximate Algorithm} we present the approximate algorithm for CVC. In Section~\ref{sec:AppTSSVDS}, we give the approximation algorithms for TSS and VDS.

\section{Preliminaries}\label{Preliminaries}
\subsection{Basic Notations}
We denote an undirected simple graph by $G=(V,E)$, where $V=[n]$ for some $n\in\mathbb{N}$ and $E\subseteq \binom{V}{2}$. Let $V(G)=V$ and $E(G)=E$ be its vertex set and edge set. For any vertex subset $S\subseteq V$, let the induced subgraph of $S$ be $G[S]$. The edges of $G[S]$ are $E[S]=E(G)\cap\binom{S}{2}$. For any $S_1,S_2\subseteq V$, we use $E[S_1,S_2]$ to denote the edge set between $S_1$ and $S_2$, i.e. $E[S_1,S_2]=\{e=(u,v)\in E\mid u\in S_1,v\in S_2\}$. For every $v\in V(G)$, we use $N(v)$ to denote the neighbors of $v$, and $d(v):=|N(v)|$.

For an orientation $O$ of a graph $G$, which can be regarded as a directed graph whose underlying undirected graph is $G$, we use $D_O^+(v)$ to denote the outdegree of $v$ and $D_O^-(v)$ its indegree. In a directed graph or an orientation, an edge $(u,v)$ is said to start at $u$ and sink at $v$. Reversing an edge is an operation, in which an edge $(u,v)$ is replaced by $(v,u)$.

In a graph $G=(V,E)$, a separator is a vertex set $X$ such that $G[V\setminus X]$ is not a connected graph. In this case we say $X$ separates $V$ into disconnected parts $C_1,C_2,...\subseteq V\setminus X$, where $C_i$ and $C_j$ are disconnected for all $i\neq j$ in $G[V\setminus X]$. 

Let $f:A\rightarrow B$ be a mapping. For a subset $A'\subseteq A$, let $f[A']$ denote the mapping with domain $A'$ and $f[A'](a)=f(a)$, for all $a\in A'$. Let $f\setminus a$ be $f[A\setminus\{a\}]$. For all $b\in B$, let $f^{-1}(b)$ be the set $\{a\in A' \mid f(a)=b\}$.

Let $\gamma\geq 0$ be a small value, we use $\mathbb{N}_\gamma$ to denote $\{0\}\cup \{(1+\gamma)^x \mid x\in \mathbb{N}\}$. For $a,b\in \mathbb{R}$, we use $a\sim_\gamma b$ to denote that $b/(1+\gamma)\leq a\leq (1+\gamma)b$. It's easy to see this is a symmetric relation. Further, we use $[a]_\gamma$ to denote $\max_{x\in \mathbb{N}_{\gamma},x\leq a} x$. Notice that $[a]_\gamma\sim_\gamma a$.

\subsection{Problems}

 \noindent\textbf{Capacitated Vertex Cover:} An instance of CVC consists of a graph $G=(V,E)$ and a capacity function $c:V\rightarrow \mathbb{N}$ on the vertices. A solution is a pair $(S,M)$ where $S\subseteq V$ and $M:E\rightarrow S$ is a mapping. If for all $v\in S,|M^{-1}(v)|\leq c(v)$ and for all $e\in E,M(e)\in e$, then we say that $S$ is feasible. The size of a feasible solution is $|S|$. The goal of CVC is to find a feasible solution of minimum size.
An equivalent description of this problem is the following. Let $O$ be an orientation of all the edges in $E$. $O$ is a feasible solution if and only if for all $v\in V,D^-_O(v)\leq c(v)$. The size of $O$ is defined as $|\{v\in V \mid d^-(v)>0\}|$. Here we actually use a directed edge $( u,v)$ to represent that $\{u,v\}$ is covered by $v$. We mainly use this definition as it's more convenient for organizing our proof and analysis. 

\medskip

\noindent\textbf{Target Set Selection:}
 Given a graph $G=(V,E)$, a threshold function $t:V\rightarrow \mathbb{N}$, and a set $S\subseteq V$, the set $S'\subseteq V$ which contains the vertices activated by $S$ is the smallest set that:
\begin{itemize}
    \item $S\subseteq S'$;
    \item For a vertex $v$, if $|N(V)\cap S'|\geq t(v)$, then $v\in S'$.
\end{itemize}
One can find the vertices activated by $S$ in polynomial time. Just start from $S':=S$, as long as there exists a vertex $v$ such that $|N(v)\cap S'|\geq t(v)$, add $v$ to $S'$, until no such vertex exists. A vertex set that can activate all vertices in $V$ is called a target set. The goal of TSS is to find a target set of minimum size.
\medskip

\noindent\textbf{Vector Dominating Set:} 
Given a graph $G=(V,E)$, a threshold function $t:V\rightarrow \mathbb{N}$, the goal of Vector Dominating Set problem is to find a  minimum vertex subset $S\subseteq V$ such that every vertex $v\in V\setminus S$ satisfies $|N(v)\cap S|\geq t(v)$.

\subsection{Tree Decomposition}
In this paper, we consider   problems parameterized by the treewidth of the input graphs. A  tree decomposition of a graph $G$ is a pair $(T,\mathcal{X})$ such that
\begin{itemize}
    \item $T$ is a rooted tree and $\mathcal{X}=\{X_\alpha :\alpha\in V(T),X_\alpha\subseteq V(G)\}$ is a collection of subsets of $V(G)$;
    \item $\bigcup_{\alpha\in V(T)}X_\alpha=V(G)$;
    \item For every edge $e$ of $G$, there exists an $\alpha\in V(T)$ such that $e\subseteq X_\alpha$;
    \item For every vertex $v$ of $G$, the set $\{\alpha\in V(T)\mid v\in X_\alpha\}$ forms a subtree of $T$.
\end{itemize}
 The width of a tree decomposition $(T,\mathcal{X})$  is $\max_{\alpha\in V(T)}{|X_\alpha|-1}$. The treewidth of a graph $G$ is the minimum width over all its tree decompositions. 
 
The sets in $\mathcal{X}$ are called ``bags''. For a node $\alpha\in V(T)$, let $T_\alpha$ denote the subtree of $T$ rooted by $\alpha$. Let $V_\alpha\subseteq V$ denote the vertex set $V_\alpha=\cup_{\alpha'\in V(T_\alpha)} X_{\alpha'}$. Let $Y_\alpha:=V_\alpha\setminus X_\alpha$. 
For a node $\alpha$, we use $\alpha_1(,\alpha_2)$ to denote its possible children. By the definition of tree decompositions, for a join node $\alpha$, $Y_{\alpha_1}\cap Y_{\alpha_2}=\emptyset$.

It is convenient to work on a \emph{nice tree decomposition}.  
Every node $\alpha\in V(T)$ in this nice tree decomposition is expected to be one of the following:
\begin{romanenumerate}
    \item \textbf{Leaf Node:} $\alpha$ is a leaf and $X_\alpha=\emptyset$;
    \item \textbf{Introducing $v$ Node:} $\alpha$ has exactly one child $\alpha_1$, $v\notin X_{\alpha_1}$ and 
 $X_\alpha=X_{\alpha_1}\cup \{v\}$;
    \item \textbf{Forgetting $v$ Node:} $\alpha$ has exactly one child $\alpha_1$, $v\notin X_{\alpha}$ and $X_\alpha\cup \{v\}=X_{\alpha_1}$;
    \item \textbf{Join Node:} $\alpha$ has exactly two children $\alpha_1$, $\alpha_2$ and $X_{\alpha}=X_{\alpha_1}=X_{\alpha_2}$.
\end{romanenumerate}

Treewidth is a popular parameter to consider because tree decompositions with optimal or approximate treewidth can be efficiently computed~\cite{korhonen2022}. We refer the reader to~\cite{cygan2015parameterized,bodlaender1994tourist,bodlaender1993tw} for more details of treewidth and nice tree decomposition. 
Using the tree balancing technique~\cite{BodlaenderH98} and the method of introducing new nodes, we can transform any tree decomposition with width $w$ in polynomial time into a nice tree decomposition with width $O(w)$,  depth upper bounded by $O(w^2\log n)$, and containing at most $O(nw)$ nodes.
Moreover, we can add $O(w)$ nodes so that the root $\alpha_0$ is assigned with an empty set $X_{\alpha_0}=\emptyset$. Notice that in this case, $Y_{\alpha_0}=V_{\alpha_0}=V$. We assume all the nice tree decompositions discussed in this paper satisfy these properties.



\section{Exact Algorithm for CVC}\label{Exact Algorithm}
We present the exact algorithm for two reasons. The first is that one can gain some basic insights on the structure of the approximate algorithm by understanding the exact algorithm, which is more comprehensible. The other is that we need to compare the intermediate results of the exact algorithm and the approximate algorithm, so the total description of the algorithm can also be regarded as a recursive definition of the intermediate results (which are the sets $\rcd$'s defined in the following).

\subsection{Definition of the Tables}

Given a tree decomposition $(T,\mathcal{X})$, we run a classical bottom-up dynamic program to solve CVC. That is, on each node $\alpha$ we allocate a record set $\rcd$. $\rcd$ contains records of the form $(\tod,k)$. A record $(\tod,k)$ consists of two elements: a mapping $\tod:X_\alpha\rightarrow \mathbb{N}$ and an integer $k\in \mathbb{N}$.  
At first, we  present a  definition of $\rcd$ by its properties. 
Then we define $\rcd$ according to the \textbf{Recursive Rules}. If our goal is only to design an exact algorithm for CVC, then there could be many different definitions of the tables which all work. However, here our definitions are elaborated so that they fit in our analysis of the approximation algorithm.
After these definitions are given, later in Theorem \ref{thm:Rrule} we prove that these two definitions coincide.

Let $G_\alpha$ denote the graph with vertex set $V_\alpha$ and edge set $E[V_\alpha]\setminus E[X_\alpha]$. We expect that the table $\rcd$  has the following properties. 
\subsubsection{Expected Properties for $\rcd$}
A record $(\tod,k)\in \rcd$ if and only if there exists $O$, an orientation of $G_\alpha$, such that
\begin{bracketenumerate}
    \item For each $v\in X_\alpha$, $\tod(v)=D_O^+(v)$ is just its out degree;
    \item $D^-_O(v)\leq c(v)$ for all $v\in Y_\alpha$;
    \item $|\{v\in Y_\alpha \mid D_O^-(v)>0\}|\leq k\leq |Y_\alpha|$.
\end{bracketenumerate}
Intuitively, $(\tod,k)\in \rcd$ if there exists a vertex set $S\subseteq Y_\alpha$ and a mapping $M:E[V_\alpha]\setminus E[X_\alpha]\rightarrow S\cup X_\alpha$ such that
\begin{itemize}
\item all edges are covered correctly, i.e. $M(e)\in e$ for all $e\in E[V_\alpha]\setminus E[X_\alpha]$;
    \item for each $v\in X_\alpha$, there are $\tod(v)$ edges from $v$ to $Y_\alpha$ that are covered by $S$, i.e. $|E[\{v\},Y_\alpha]\cap \cup_{u\in S} M^{-1}(u)|=\tod(v)$;  
    \item $M$ satisfies the capacity constraints for vertices in $Y_\alpha$, i.e. for all $v\in Y_\alpha$, $|M^{-1}(v)|\leq c(v)$;
\item $|S|\le k\le |Y_\alpha|$.    
\end{itemize}
One can imagine that $S$ is a feasible solution for a spanning subgraph of $G_\alpha$, where the vector $\tod$ governs the edges between $X_\alpha$ and $Y_\alpha$.

Note that the root node $\alpha_0$ satisfies $X_{\alpha_0}=\emptyset$, and $G_{\alpha_0} = G$. So if $\rcds{0}$ is correctly computed, then the $k$ values in those records in $\rcds{0}$ have a one-to-one correspondence to all feasible solution sizes for the original instance. We output $\min_{(\tod,k)\in \rcds{0}} k$ to solve the instance.

\subsubsection{Recursive Rules for $\rcd$}
Fix a node $\alpha \in V(T)$, if $\alpha$ is a introducing node or a forgetting node, let $\alpha_1$ be its child. If $\alpha$ is a join node, let $\alpha_1,\alpha_2$ be its children. 
In case $\alpha$ is a:
\begin{description}
\item[Leaf Node.] $\rcd=\{(\tod,k)\}$, in which $\tod$ is a mapping with empty domain and $k:=0$.

 \item[Introducing $v$ Node.] Note that by the properties of tree decompositions, there is no edge between $v$ and $Y_\alpha$ in $G$.
A record $(\tod,k)\in \rcd$ if and only if $(\tod\setminus v,k)\in \rcds{1}$ and $\tod(v)=0$. 

\item[Join Node.] $(\tod,k)\in \rcd$ if and only if there exist $(\tod_1,k_1)\in \rcds{1}$ and $(\tod_2,k_2)\in \rcds{2}$ such that for all $ v\in X_\alpha$, $\tod(v)=\tod_1(v)+\tod_2(v)$ and $k=k_1+k_2$.

\item[Forgetting $v$ Node.]$(\tod,k)\in \rcd$ if and only if there exists $(\tod_1,k_1)\in \rcds{1}$ satisfying one of the following conditions:
\begin{bracketenumerate}
    \item $k_1=k$, $\tod_1(v)=|N(v)\cap Y_\alpha|$ and $\tod_1\setminus v=\tod$. In this case,  $v$ is not ``included in $S$''. All the edges between $v$ and $Y_\alpha$ must be covered by other vertices in $Y_\alpha$.
    \item $k_1=k-1$ and there exist $ \Delta(v)\subseteq N(v)\cap X_\alpha$ and $A\in [|N(v)\cap Y_\alpha|-c(v)+|\Delta(v)|,|N(v)\cap Y_\alpha|]$ such that $\tod_1(v) = A$, $\tod_1(u)=\tod(u)-1$ for all $u\in \Delta(v)$, and $\tod_1(u)=\tod(u)$ for all $u\in X_{\alpha_1}\setminus (\Delta(v)\cup \{v\})$. In this case, $v$ is ``included in $S$''. We enumerate a set $\Delta(v)\subseteq N(v)\cap X_\alpha$ of edges between $v$ and $X_\alpha$ and let $v$ cover these edges. Note that for a record $(\tod_1,k_1)\in \rcds{1}$, there are $|N(v)\cap Y_\alpha|-\tod_1(v)$ edges that are covered by $v$. To construct $(\tod,k)$ from $(\tod_1,k_1)$, we need to check that $c(v)\ge |\Delta(v)|+|N(v)\cap Y_\alpha|-\tod_1(v)$, which is implicitly done by the setting $\tod_1(v)=A\geq |N(v)\cap Y_\alpha|-c(v)+|\Delta(v)|$.
\end{bracketenumerate}
\end{description}
\begin{remark}
In fact, one can find many different ways to define the dynamic programming table for CVC. We use this definition because we want to upper bound the values of records in $\rcd$ by the size of solution (Lemma \ref{lb}), so we need to record ``outdegrees'' rather than ``indegrees'' or ``capacities''. 
\end{remark}

\noindent\textbf{Valid certificate.}  Notice that all the rules are of the form $(\tod_1,k_1)\in \rcds{1}\Rightarrow (\tod,k)\in \rcd$ or $(\tod_1,k_1)\in \rcds{1}\wedge(\tod_2,k_2)\in \rcds{2}\Rightarrow (\tod,k)\in \rcd$, thus a rule can actually be divided in to two parts: we found a ``valid certificate'' $(\tod_1,k_1)\in \rcds{1}$ (and $(\tod_2,k_2)\in \rcds{2}$, for join nodes), then we add a ``product'' $(d,k)\in \rcd$ based on the certificate. 
In fact, every record in $\rcds{1}$ can be a valid certificate in introducing nodes, and every pair of records $((\tod_1,k_1),(\tod_2,k_2))\in \rcds{1}\times \rcds{2}$ can be a valid certificate in join nodes. But in forgetting $v$ nodes, we further require that $\tod_1(v)$ satisfies some condition. To be specific, in a forgetting node $\alpha_1$  we say $(\tod_1,k_1)\in \rcds{1}$ is {valid} if it satisfies the following condition:


\begin{itemize}
    \item[($\star$)] $\tod_1(v)=|N(v)\cap Y_{\alpha}|$ or $\geq|N(v)\cap Y_\alpha|-c(v)+|\Delta(v)|$ for some $\Delta(v)\subseteq N(v)\cap X_\alpha$.
\end{itemize} 


\begin{theorem}\label{thm:Rrule}
 The set $\{\rcd : \alpha\in V(T)\}$ can be computed by the recursive rules above in time $n^{w+O(1)}$, and the \textbf{Expected Properties} are satisfied. 
\end{theorem}
The proof sketch of the correctness of these rules are presented in Appendix \ref{proofexact}. As $|\rcd|\leq n^{w+2}$ for all $\alpha\in V(T)$ and the enumerating $\Delta(v)$ procedure in dealing with a forgetting node runs in time $w^{O(w)}$, it's not hard to see that this algorithm runs in time $n^{w+O(1)}$ (for $w$ small enough compared to $n$).

\section{Approximation Algorithm for CVC}\label{Approximate Algorithm}
Let $\epsilon$ be a small value to be determined later. We try to compute an approximate record set $\rcdap$ for each node $\alpha$, still using bottom-up dynamic programming like what we do in the exact algorithm. An approximate record is a pair $(\hat{\tod},\hat{k})$, where $\hat{k}\in \mathbb{N}$ and $\hat{\tod}$ is a mapping from $X_\alpha$ to $\mathbb{N}_{\epsilon}=\{0\}\cup\{(1+\epsilon)^x \mid x\in \mathbb{N}\}$. As we can see, $\hat{\tod}$ can take non-integer values.

\textbf{Height of a Node}
The height $h$ of a node $\alpha$ is defined by the length of the longest path from $\alpha$ to a leaf which is descendent to $\alpha$.
By this definition, the height of a node is $1$ plus the maximum height among the heights of its children.
Let the height of the root node be $h_0$. According to the property of nice tree decompositions, $h_0$ is at most $O(w^2\log n)$.

Let $\epsilon_h,\delta_h$ be two non-negative values (which are functions of $h,n$ and $w$) to be determined later.

\textbf{$h$-close records.} If an exact record $(\tod,k)$ and an approximate record $(\hat{\tod},\hat{k})$ satisfy $\tod(v)\sim_{\epsilon_h} \hat{\tod}(v)$ for all $v\in X_\alpha$ and $k\sim_{\delta_h}\hat{k}$, then we say these two records are $h$-close. 

We expect that for each node $\alpha$, $\rcdap$ satisfies the following. Let the height of $\alpha$ be $h$.
\begin{enumerate}[(A)]
    \item If $(\tod,k)\in \rcd$, then there exists $(\hat{\tod},\hat{k})\in \rcdap$ which is $h$-close to $(\tod,k)$.
    \item If $(\hat{\tod},\hat{k})\in \rcdap$, then there exists $(\tod,k)\in \rcd$ which is $h$-close to $(\hat{\tod},\hat{k})$.
\end{enumerate}
After $\rcdaps{0}$ is correctly computed (i.e. satisfying (A) and (B)), we output the value $\hat{k}_{\min}=(1+\delta_{h_0})\min_{(\hat{\tod},\hat{k})\in \rcdaps{0}}\hat{k}$. Let $OPT$ be the size of the minimum solution, which equals to $\min_{(\tod,k)\in \rcds{0}} k$. We claim that $\hat{k}_{\min}\in [OPT,(1+\delta_{h_0})^2 OPT]$.
\begin{proof}
    By property (B), we have $OPT\leq (1+\delta_{h_0})\min_{(\hat{\tod},\hat{k})\in \rcdaps{0}}\hat{k}$. By property (A), we have $\min_{(\hat{\tod},\hat{k})\in \rcdaps{0}}\hat{k}\leq (1+\delta_{h_0})OPT$. The claim follows by combining these two inequalities.
\end{proof}
We need the following procedure to  test in polynomial time if a sub-problem is solvable when we are allowed to use all vertices to cover the edges.
\begin{lemma}\label{pl}
    Testing whether $(\tod,|Y_\alpha|)\in \rcd$ for any $\tod$ can be done in $n^{O(1)}$ time.
\end{lemma}
\begin{proof}
    Construct a directed graph with vertex set $\{s,t\}\cup (E[V_\alpha]\setminus E[X_\alpha])\cup V_\alpha$. 
    For each $e\in (E[V_\alpha]\setminus E[X_\alpha])$ add an edge $(s,e)$ with capacity $1$. For each $e=(u,v)\in (E[Y_\alpha]\setminus E[X_\alpha])$ add an edge $(e,u)$ and an edge $(e,v)$ both with capacity $1$. For each $v\in X_\alpha$ add an edge $(v,t)$ with capacity $|N(v)\cap Y_{\alpha}|-\tod(v)$. For each $v\in Y_\alpha$ add an edge $(v,t)$ with capacity $c(v)$. We claim that $(\tod,|Y_\alpha|)\in \rcd$ if and only if there is a flow from $s$ to $t$ with value $|E[V_\alpha]\setminus E[X_\alpha]|$.  For the 'if' part, notice that by the well-known integrality theorem for network flow, there exists a integral flow with the same value. Every integral flow with this value can be transform to an $O$ as expected in the \textbf{Expected Properties}: An edge $e\in E[Y_\alpha]\setminus E[X_\alpha]$ is oriented so that it sinks at vertex $v$ if $(e,v)$ has flow value $1$, then for each vertex $v\in X_\alpha$, reverse some edges in $E[\{v\},Y_\alpha]$ so that $D_O^+(v)=\tod(v)$, if the flow carried in $(v,t)$ is less than $|N(v)\cap Y_{\alpha}|-\tod(v)$. One can construct a flow with the value based on an orientation $O$, too. Thus the 'only if' part is easy to see, too.
\end{proof}

We first define $\{\rcdap:\alpha\in V(T)\}$ using the following \textbf{Recursive Rules}. Then we prove that these sets satisfy the properties (A) and (B). The basic idea of our approximate algorithm is to run the exact algorithm in an ``approximate way''. 
For a rule formed as $(\hat{\tod}_1,\hat{k}_1)\in \rcdaps{1}\Rightarrow (\hat{\tod},\hat{k})\in \rcdap$ or $(\hat{\tod}_1,\hat{k}_1)\in \rcdaps{1}\wedge (\hat{\tod}_2,\hat{k}_2)\in \rcdaps{2}\Rightarrow (\hat{\tod},\hat{k})\in \rcdap$, we also call $(\hat{\tod}_1,\hat{k}_1)$ (and $(\hat{\tod}_2,\hat{k}_2)$) the certificate while $(\hat{\tod},\hat{k})$ is the product.

\subsection{Recursive Rules for $\rcdap$}

Fix a node $\alpha\in V(T)$ with height $h$, in case $\alpha$ is a:
\begin{description}
    
\item[Leaf Node.] $\rcdap=\{(\hat{\tod},\hat{k})\}$, in which $\hat{\tod}$ is a mapping with empty domain and $\hat{k}=0$.

\item[Introducing $v$ Node.]

A record $(\hat{\tod},\hat{k})\in \rcdap$ if and only if $(\hat{\tod}\setminus v,\hat{k})\in \rcdaps{1}$ and $\hat{\tod}(v)=0$.

\item[Join Node.] $(\hat{\tod},\hat{k})\in \rcdap$ if and only if there exists $(\hat{\tod}_1,\hat{k}_1)\in \rcdaps{1},(\hat{\tod}_2,\hat{k}_2)\in \rcdaps{2}$ such that for each $v\in X_\alpha$, $\hat{\tod}(v)=[\hat{\tod}_1(v)+\hat{\tod}_2(v)]_\epsilon$ and $\hat{k}=\hat{k}_1+\hat{k}_2$.

\item[{Forgetting $v$ Node}.] This case is the most complicated. Think in this way: we pick $(\hat{\tod}_1,\hat{k}_1)\in \rcdaps{1}$ and based on it we try to construct $(\hat{\tod},\hat{k})$ to add into $\rcdap$. Notice that in the exact algorithm, not every $(\tod_1,k_1)\in \rcds{1}$ can be used to generate a corresponding product $(\tod,k)\in \rcd$ --- it has to be the case that $\tod_1(v)=|N(v)\cap Y_\alpha|$ or $\tod_1(v)\geq|N(v)\cap Y_\alpha|-c(v)+|\Delta(v)|$, which is what we called to be a valid certificate. We have to test both the validity of the certificate and its exact counterpart using an indirect way. 
So there are three issues we need to address:
\begin{enumerate}[(a)]
    \item The requirement for $(\hat{\tod}_1,\hat{k}_1)$ being valid, i.e. satisfying the ``approximate version'' of condition $(\star)$;
    \item There exists a valid exact counterpart $(\tod_1,k_1)\in\rcds{1}$ of $(\hat{\tod}_1,\hat{k}_1)$ satisfying condition $(\star)$;
    \item How to construct $(\hat{\tod},\hat{k})$.
\end{enumerate}

(b) seems impossible since we do not compute $\rcds{1}$, we obtain this indirectly using Lemma~\ref{pl}. Later we explain why such an approach reaches our goal. Formally, suppose we have $(\hat{\tod}_1,\hat{k}_1)\in \rcdaps{1}$, we consider two cases:
\begin{enumerate}[(1)]
    \item $v$ is not ``included''.
    \begin{description}
        \item (1a) See if $\hat{\tod}_1(v)\sim_{\epsilon_{h-1}}|N(v)\cap Y_\alpha|$;
        \item (1b) See if $(\tod_t,|Y_{\alpha_1}|)\in \rcds{1}$, where
        $\tod_t(u)=\lceil\hat{\tod}_1(u)/(1+\epsilon_{h-1})\rceil$ for all $u\in X_{\alpha_1}\setminus \{v\}$ and $\tod_t(v)=|N(v)\cap Y_\alpha|$  (This is polynomial-time tractable by Lemma \ref{pl});
        \item (1c) If (a) and (b) are satisfied, then add $(\hat{\tod},\hat{k})$ to $\rcdap$, where $\hat{\tod}=\hat{\tod}_1\setminus v,\hat{k}=\hat{k}_1$.
    \end{description}
    \item $v$ is ``included''. We enumerate $\Delta(v)\subseteq N(v)\cap X_\alpha$ and integer $A$ satisfying $A\in [|N(v)\cap Y_\alpha|-c(v)+|\Delta(v)|,|N(v)\cap Y_\alpha|]$.
    \begin{description}
        \item (2a) See if $\hat{\tod}_1(v)\geq A/(1+\epsilon_{h-1})$;
        \item (2b) See if $(\tod_t,|Y_{\alpha_1}|)\in \rcds{1}$, where $\tod_t(u)=\lceil \hat{\tod}_1(u)/(1+\epsilon_{h-1}) \rceil$ for all $u\in X_{\alpha_1}\setminus \{v\}$ and $\tod_t(v)=A$ (By Lemma \ref{pl}, this is still polynomial-time tractable);
        \item (2c) If (a) and (b) are satisfied, then add $(\hat{\tod},\hat{k})$ to $\rcdap$, where $\hat{\tod}(u)=\hat{\tod}_1(u)$ for all $u\in X_\alpha\setminus \Delta(v)$, $\hat{\tod}(u)=[\hat{\tod}(u)+1]_\epsilon$ for all $u\in \Delta(v)$, $\hat{k}=\hat{k}_1+1$.
    \end{description}
\end{enumerate}
\end{description}

\begin{theorem}\label{correctness}
    Set $\epsilon = \frac{1}{(w^2\log n)^3}$, $\epsilon_h=2h\epsilon$ and $\delta_{h}= 4(h+1)h\epsilon$. Suppose $n$ is great enough. When the dynamic programming is done, for all $\alpha$, $\rcdap$ satisfies property (A) and (B).
\end{theorem}

\begin{proof}
    \textbf{(of Theorem \ref{main})} According to Theorem \ref{correctness} and the above discussion, we immediately get $\hat{k}_{\min}\in [OPT,(1+\delta_{h_0})^2OPT]$. By the property of nice tree decomposition, $h_0$ is at most $O(w^2\log n)$, thus $\hat{k}_{\min}\in [OPT,(1+O(1/(w^2\log n)))^2 OPT]=[OPT,(1+O(1/(w^2\log n))) OPT]$.

    The space we need to memorize each $\rcdap$ is $O((w^6\log ^4n)^w n^{O(1)})$. Computing a leaf/introduce/join node we need $O((w^6\log ^4n)^{2w} n^{O(1)})$ time. In a forgetting node, we may need to enumerate some set $\Delta(v)\subseteq N(v)\cap X_\alpha$, which requires time $O(2^{|X_\alpha|})=O(2^{w+1})$. So computing a Forgetting node requires $O((w^6\log ^4n)^w 2^w n^{O(1)})$ time. The tree size is polynomial, so the total running time is FPT.
\end{proof}

To prove Theorem~\ref{correctness}, we need a few lemmas. The proof of Lemma~\ref{monotone} and Lemma~\ref{error} are presented in Appendix~\ref{proofapprx}. 
Lemma~\ref{monotone} and Lemma~\ref{error} are some simple observations. To understand why we need Lemma~\ref{lb} and Lemma~\ref{extent}, remember that we have a complicated recursive rule for forgetting nodes in which we verifies (a) and (b). However, we cannot directly verify if a valid exact record described in (b) exists, because we don't have $\rcds{1}$. We overcome this by verifies a feasible partial solution (e.g. $(\tod_1,|Y_{\alpha_1}|)$ in (1b)) rather than an optimal one, which can be done by Lemma~\ref{pl}. When we are computing $\rcdap$, we assume that $\rcdaps{1}$ has been correctly computed, i.e. it satisfies (A) and (B). 
So there exists $(\tod_1,k_1)\in rcds{1}$ which is $h-1$-close to $(\hat{\tod}_1,\hat{k}_1)$. However, we don't know if $(\tod_1,k_1)$ is a so-called valid certificate. Lemma~\ref{extent} shows how to modify $(\tod_1,k_1)$ so that it becomes we want in (b), knowing $(\tod_t,|Y_{\alpha_1}|)\in \rcds{1}$. We introduces some error like $o(1)\tod_1(v)$ on $k_1$ in this procedure. Lemma~\ref{lb} helps us to rewrite it as $o(1)k_1$.
\begin{lemma}\label{monotone}
    If $(\tod,k)\in \rcd$ for some node $\alpha$, then for every $({\tod}',k')$ with $\tod(v)\geq {\tod}'(v)$ for all $ v\in X_\alpha$ and $k'$ satisfying $k\leq k'\leq |Y_\alpha|$, we have $({\tod}',k')\in \rcd$.
\end{lemma}

\begin{lemma}\label{error}
    Let $a,b,a',b'\in \mathbb{R},h\in \mathbb{N}^+$, $\epsilon_h\in [0,1]$, $a'\sim_{\epsilon_h} a$ and $b'\sim_{\epsilon_h} b$. Then we have $[a'+b']_\epsilon\sim_{\epsilon_{h+1}} (a+b)$.
\end{lemma}

\begin{lemma}\label{lb}
    For all $(\tod,k)\in \rcd$ and $ v\in X_\alpha,k\geq \tod(v)$.
\end{lemma}

\begin{proof}
    Let $O$ be the orientation. Let $N^+(v)=\{u\in V(G) : (v,u)\in E(G)\}$ be out neighbors of $v$. By definition, we have $\tod(v)=|N^+(v)|\leq |\{u\in Y_\alpha \mid D_O^-(u)>0\}|\leq k$. 
\end{proof}

\begin{lemma}\label{extent}
    Fix some $(\tod,k)\in \rcd,v\in X_\alpha$ and some integer $p> 0$ satisfying $k+p\leq |Y_\alpha|$. Let $\tod_{m}:X_\alpha\rightarrow \mathbb{N}$ be a function such that $\tod_m(v)=\tod(v)+p$ and $\tod_m\setminus v= \tod\setminus v$.   We have $(\tod_m,|Y_\alpha|)\in \rcd$ if and only if $(\tod_m,k+p)\in \rcd$.
\end{lemma}

\begin{proof}
    On one hand, the 'if' part is obvious by Lemma \ref{monotone}. On the other hand, we prove that  $(\tod_m,|Y_\alpha|)\in \rcd$ implies $({\tod}',k+1)\in \rcd$, where ${\tod}'(v)=\tod(v)+1,{\tod}'\setminus v=\tod\setminus v$. Then we can repeatedly increase the value of $k$ by $1$ for $p$ times to obtain the 'only if' part.
    Let the orientation corresponding to $(\tod,k)$ and $(\tod_m,|Y_\alpha|)$ be $O_1,O_2$ respectively. Now let $G'$ be a graph with vertex set $Y_\alpha\cup \{v\}$. A directed edge $(x,y)$ is in $G'$ if and only if $(x,y)\in O_2$ and $(y,x)\in O_1$. 
    
    By picking $O_1$ so that the number of edges in $G'$ is minimized, we can assume that $G'$ contains no cycle. Otherwise if $G'$ contains a cycle, we can reverse every edge along the cycle in $O_1$ so that it is still a valid orientation for $(\tod,k)$ but the number of edges in $G'$ decreases. 
    
    As $D_{O_2}^+(v)> D_{O_1}^+(v)$, there exists an non-empty path in $G'$ starting from $v$ ending at, say, $v'\neq v$ such that $v'$ has no out edge in $G'$. This implies $D_{O_1}^-(v')\leq D_{O_2}^-(v')-1$, or $v'$ will have an out edge in $G'$. We reverse the edges along this path in $O_1$. Let the new orientation be $O_3$. $D_{O_3}^-(v')\leq D_{O_1}^-(v')+1\leq D_{O_2}^-(v')\leq c(v)$. Moreover, $\{u \mid D_{O_3}^-(u)>0\}\setminus \{u \mid D_{O_1}^-(u)>0\}\subseteq \{v'\}$. Thus, $O_3$ is a valid orientation for $({\tod}',k+1)$.
\end{proof}

\subsection{Theorem~\ref{correctness} Proof Sketch}

Due to space limit, the complete proof is presented in Appendix~\ref{proofapprx}.

We use induction on nodes, following a bottom-up order on the tree decomposition. Leaf nodes satisfy property (A) and (B), because $\rcd=\rcdap$ for every leaf node. Fix a node $\alpha$ of height $h$, by induction, we assume that every node descendent to $\alpha$ satisfies (A) and (B). We only need to prove that $\alpha$ satisfies both (A) and (B).  
We make a case distinction based on the type of $\alpha$. The case where $\alpha$ is a forgetting node is the most complicated and requires lemma \ref{lb} and \ref{extent}. The other two types follow Lampis' framework.

To show $\alpha$ satisfies (A), we need to prove the existence of some $(\hat{\tod},\hat{k})\in \rcdap$ for any given $(\tod,k)\in \rcd$ such that $(\hat{\tod},\hat{k})$ and $(\tod,k)$ are $h$-close. This is done by first picking up the certificate of $(\tod,k)$, that is, the record $(\tod_1,k_1)\in \rcds{1}$ (or a pair of records in the case $\alpha$ is a join node, we omit join node case in the following sketch) which ``produces'' $(\tod,k)$ based on recursive rules for $\rcd$. Then by induction hypothesis, there is an $(h-1)$-close record $(\hat{\tod}_1,\hat{k}_1)$  in $\rcdaps{1}$. If $\alpha$ is not a forgetting node, then according to recursive rules for $\rcdap$, there exists $(\hat{\tod},\hat{k})\in \rcdap$. We prove that $(\hat{\tod},\hat{k})$ and $(\tod,k)$ are $h$-close. If $\alpha$ is a forgetting node, then we verify (1b) or (2b) by applying Lemma~\ref{monotone} on $(\tod_1,k_1)$.

To show $\alpha$ satisfies (B), if $\alpha$ is not a forgetting node, then we pick up and compare some records in a different order: We start from $(\hat{\tod},\hat{k})\in \rcdap$; Then we pick $(\hat{\tod}_1,\hat{k}_1)\in \rcdaps{1}$ according to recursive rules for $\rcdap$; Next we pick $(\tod_1,k_1)\in \rcds{1}$ based on induction hypothesis; Finally we find out $(\tod,k)\in \rcd$ using recursive rules for $\rcd$.
If $\alpha$ is a forgetting node, suppose the record $(\hat{\tod},\hat{k})\in \rcdap$ is produced by $(\hat{\tod}_1,\hat{k}_1)$. The main idea is to apply Lemma~\ref{extent} on $(\tod_t,|Y_{\alpha_1}|)$, the record verified by (1b) or (2b), and $(\tod_1,k_1)$, the record $(h-1)$-close to $(\hat{\tod}_1,\hat{k}_1)$, so as to show the existence of some $(\tod,k)\in \rcd$. At the same time we use Lemma~\ref{lb} to bound $k$.

\section{Approximation algorithms for TSS and VDS}\label{sec:AppTSSVDS}
In this section, we introduce the \emph{vertex subset problem} which is a generalization of many graph problems. Then we present a sufficient condition for the existence of parameterized approximation algorithms for such problems parameterized by the treewidth. Finally, we apply our algorithm to target set selection problem (TSS) and vector dominating set problem (VDS), which are both vertex subset problems satisfying this condition.  The definitions bellow are inspired by Fomin, et al. \cite{monotonelocalsearch}.

\begin{definition}[{Vertex Subset Problem}]
    A vertex subset problem $\Phi$ takes a string $I\in \{0,1\}^*$ as an input, which encodes a graph $G_I=(V_I,E_I)$ and some possible additional information, e.g. threshold values on vertices. $\Phi$ is identified by a function $\mathcal{F}_\Phi$ which maps a string $I\in \{0,1\}^*$ to a family of vertex subsets of $V_I$, say $\mathcal{F}_\Phi(I)\subseteq 2^{V_I}$. Any vertex set in $\mathcal{F}_\Phi(I)$ is a \textbf{solution} of the instance $I$. The goal is to find a minimum sized solution.
\end{definition}

We will often select a set of vertices and assume that it is included in a solution, and then consider the remaining sub-problem. So we introduce the concept of partial instances.
\begin{definition}[{Partial Vertex Subset Problem}]
    Let $\Phi$ be a vertex subset problem. The partial version of $\Phi$ takes a string $I\in \{0,1\}^*$ appended with a vertex subset $U\subseteq V_I$ as input. We call such a pair $(I,U)$ a partial instance of $\Phi$. Any vertex set $W\subseteq V_I\setminus U$ is a solution if and only if $W\cup U\in \mathcal{F}_\Phi(I)$. Still, the goal is to find a minimum sized solution.
\end{definition}
We consider the following conditions of a vertex subset problem $\Phi$.
\begin{itemize}
    \item $\Phi$ is \textbf{monotone}, if for any instance $I$, $S\in \mathcal{F}_\Phi(I)$ implies for all $S'$ satisfying $S\subseteq S'\subseteq V_I$, $S'\in \mathcal{F}_\Phi(I)$.
    \item $\Phi$ is \textbf{splittable}, if: for any instance $I$ and any separator $X$ of $G_I$ which separates $V_I\setminus X$ into disconnected parts $V_1,V_2,...,V_p$, if $S_1,S_2,...,S_p$ are vertex sets such that $S_i$ is a solution for the partial instance $(I,V_I\setminus V_i),\forall 1\leq i\leq p$, then $X\cup\bigcup_{1\leq i\leq p} S_i$ is a solution for $I$.
\end{itemize}


It is trivial to show the monotonicity for TSS and VDS. To see that they are splittable,
observe that given an instance $I=(G,t)$ of VDS for example, fix some $X\subseteq V(G)$, a set $S$ containing $X$ is a solution for $I$ if and only if $S\setminus X$ is a solution for $I'=(G',t')$, where $G'=G[V\setminus X]$ and $t'(v)=t(v)-|N(v)\cap X|$ for all $v\in V\setminus X$. If $X$ is a separator, then the graph $G'$ is not connected, and the union of any solutions of each component in $G'$, with $X$ together forms a solution of $I$. This observation also works for TSS. 

The main theorem in this section is to show the tractability, in the sense of parameterized approximation, of monotone and splittable vertex subset problems on graphs with bounded treewidth. 

\begin{theorem}
    \label{main2}
    Let $\Phi$ be a vertex selection problem which is monotone and splittable. 
    If there exists an algorithm such that on input a partial instance of $\Phi$ appended with a corresponding nice tree decomposition with width $w$, it can run in time $f(\ell,w,n)$ and
    \begin{itemize}
        \item either output the optimal solution, if the size of it is at most $\ell$;
        \item or confirm that the optimal solution size is at least $\ell+1$
    \end{itemize} 
    then there exists an approximate algorithm for $\Phi$ with ratio $1+(w+1)/(l+1)$ and runs in time $f(l,w,n)\cdot n^{O(1)}$, for all $l\in \mathbb{N}$.
\end{theorem}

We provide a trivial algorithm for the partial version of TSS. Given a partial instance $(I=(G,t),U)$, we search for a solution of size at most $\ell$ by brute-force. This takes time $f(\ell,w,n)=n^{\ell+O(1)}$. Setting $l:=C$ in Theorem \ref{main2}, we simply get the following.
\begin{corollary}\label{TSS}
    (Restated version of Theorem \ref{thm:tssapp})
    For all constant $C$, TSS admits a $1+(w+1)/(C+1)$-approximation algorithm running in time $n^{C+O(1)}$.
\end{corollary}

As mentioned before, Raman et al.\cite{vdsw1fpt} showed that VDS is $W[1]$-hard parameterized by $w$, but FPT with respect to the combined parameter $(k+w)$ where $k$ is the solution size. The running time of their algorithm is $k^{O(wk^2)}n^{O(1)}$. A partial instance $(I,U)$ of VDS can be transformed to an  equivalent VDS instance, in which the input graph is $G[V_I\setminus U_I]$, so this algorithm can also be used for the partial version of VDS.
Set $l:=w^2(\log \log n/\log \log \log n)^{0.5}$ in Theorem \ref{main2}, we get Corollary \ref{VDS}.

\begin{corollary}\label{VDS}
    (Restated version of Theorem \ref{thm:vdsapp})
    Vector Dominating Set admits a $1+1/(w\log\log n)^{\Omega(1)}$-approximation algorithm running in time $2^{O(w^5\log w \log\log n)}n^{O(1)}$.
\end{corollary}

Notice that we can't obtain a $(1+o(1))$-approximation for TSS using a similar approach, because solving TSS in $f(w+k)n^{O(1)}$-time is $W[1]$-hard~\cite{tssw1}.

One may also think of applying Theorem~\ref{main2} to CVC, since CVC is FPT when parameterized by solution size~\cite{cdscvc08}. However, CVC is not splittable. Think of a simple $3$-vertex graph with vertex set $\{a,b,c\}$ and edge set $\{\{a,b\},\{b,c\}\}$. The capacities are: $c(a)=0,c(b)=1,c(c)=0$. $\{b\}$ is a separator in this graph and empty sets are two solutions for the two partial instances. However, $\{b\}$ cannot cover both two edges in the original graph.


\subsection{The Algorithm Framework}
To prove Theorem \ref{main2}, we introduce the concept of $l$-good node.
\begin{definition}[{$l$-good Node}] Let $I$ be an instance of a vertex selection problem $\Phi$ and $(T,\mathcal{X})$ be a nice tree decomposition of any subgraph of $G_I$. A node $\alpha\in V(T)$ is an $l$-good node if the partial instance $(I,V_I\setminus Y_\alpha)$ admits a solution of size at most $l$. 
\end{definition}
For a node $\alpha$, let $N^-_\alpha$ denote the set of all children of $\alpha$. We present the pseudocode of our algorithm in Algorithm~\ref{subprocess}. Figure~\ref{fig:demAlg} in Appendix~\ref{proofTSSVDS} illustrates how the sets defined in Algorithm~\ref{subprocess} are related. Algorithm~\ref{subprocess} solves the partial version of $\Phi$. For the original version, when we get an instance $I$, we just create an equivalent partial instance $(I,\emptyset)$ appended with a nice tree decomposition $(T,\mathcal{X})$ and an integer $l$, then we run $Solve((I,\emptyset),(T,\mathcal{X}),l)$. The analyze of Algorithm~\ref{subprocess} is presented in Appendix~\ref{proofTSSVDS}. 


\noindent\textbf{Main idea of Algorithm~\ref{subprocess}:} Let $Alg$ be an algorithm solving partial instances in time $f(l,w,n)$. Given a partial instance $(I,D)$ and a nice tree decomposition $(T,\mathcal{X})$ on $G[I\setminus D]$, we run $Alg$ to test the goodness of each node. If the root node is $l$-good, then $(I,D)$ has a solution with size at most $l$, we use $Alg$ to find the optimal solution. If a leaf node is not $l$-good then by monotonicity $I$ has no solution\footnote{By our definition of vertex subset problem, the set of solutions can be empty. However any instance of TSS or VDS admits at least one solution which is the whole vertex set.}.
Otherwise, we can pick a lowest node $\alpha$ which is not $l$-good. Then all its children are $l$-good. Such a node has nice properties.
 \begin{itemize}
     \item  On one hand, since all children of $\alpha$ are $l$-good, the partial instances $(I,V_I\setminus Y_{\alpha_c})$ can be optimally solved by $Alg$ for each $\alpha_c$ a child of $\alpha$.  Adding $X_{\alpha_c}$ and the optimal solution $E_{\alpha_c}$ for $(I,V_I\setminus Y_{\alpha_c})$ into the solution enables us to ``discard'' the whole subtree rooted by $\alpha_c$ and the corresponding vertices, i.e. $V_{\alpha_c}$; 
     \item On the other hand, as $\alpha$ is not $l$-good, by the splittable and monotone properties, we can deduce that the optimal solution $S^*$ has an intersection of size at least $(l+1)$ with $Y_\alpha$ i.e. $|S^*\cap Y_\alpha|\geq l+1$. 
\end{itemize}
Based on these properties, the algorithm iteratively finds one such node $\alpha$ and includes $X_{\alpha_c}\cup E_{\alpha_c}$ for its every child $\alpha_c$ into the solution, then ``removes'' $V_{\alpha_c}$ from the graph.
Once we repeat this procedure, the optimal solution size decreases by at least $|S^*\cap (\bigcup_{\alpha_c}V_{\alpha_c})|\geq |S^*\cap Y_\alpha|\geq l+1$. For each $\alpha_c$, we use $Alg$ to find the optimal solution $E_{\alpha_c}$, so in each $Y_{\alpha_c}$ we select at most $|S^*\cap Y_{\alpha_c}|$ vertices. The ``non-optimal'' part is $\bigcup_{\alpha_c} X_{\alpha_c}$, which is at most $O(w)=O(w/l)|S^*\cap(\bigcup_{\alpha_c}V_{\alpha_c})|$.
Therefore, the approximation ratio is upper bounded by $1+\frac{|\bigcup_{\alpha_c}X_{\alpha_c}|}{l+1}\le 1+O(w/l)$.



\begin{algorithm}[ht]
    \caption{Subprocess $Solve()$}
    \label{subprocess}
    \LinesNumbered
    \KwIn{A partial instance $(I,D)$ of $\Phi$, a nice tree decomposition $(T,\mathcal{X})$ of $G_I[V_I\setminus D]$ with width $w$, $l\in \mathbb{N}$ an integer.}
    \KwOut{A solution $S$ to $(I,D)$, or 'there exists no solution'.}
    \BlankLine
    \For{each node $\alpha$}{
        Use $Alg$ to test if $\alpha$ is an $l$-good node;\\
        \If{$\alpha$ is $l$-good}{
            $E_\alpha:=$ the minimum solution for $(I,V_I\setminus Y_\alpha)$;
        }
    }
    \If{the root $\alpha_0$ is $l$-good}{
        Return $E_{\alpha_0}$;
    }

    Find a node $\alpha$ which is not $l$-good with minimum height;\\
    \If{$\alpha$ is a leaf node}{
        Return 'there exists no solution';
    }
    $E':=\emptyset$;\\
    $F:=\emptyset$;\\
    \For{each $\alpha_c\in N^-_{\alpha}$}{
        $E':=E'\cup E_{\alpha_c}\cup X_{\alpha_c}$;\\
        $F:=F\cup V_{\alpha_c}$;\\
    }
    Find a nice tree decomposition $(T',\mathcal{X}')$ for $G_I[V_I\setminus (D\cup F)]$;\\
    Return $E'\cup Solve((I,D\cup F),(T',\mathcal{X}'),l)$;
\end{algorithm}




\bibliography{lipics-v2021-sample-article}

\appendix

\section{Proof Sketch of Theorem \ref{thm:Rrule}}
\label{proofexact}

It is easy to see that $|\rcd|\leq n^{O(w)}$ for all $\alpha$. And one execution of a recursive rule takes time at most polynomial of the size of some $\rcd$. Thus the total running time is $n^{O(w)}\cdot O(w^2\log n)=n^{O(w)}$.

To prove the correctness, we need to show the record sets computed by the recursive rules satisfy the expected properties. We use induction. Leaf nodes are trivial to verify. Fix a node $\alpha$, assume that for every node descendent to it, the corresponding record set is correctly computed. The proof then contains the 'if' part and the 'only if' part. For the 'if' part we have some $O$, $(\tod,k)$ satisfying the expected properties and aim to prove $(\tod,k)$ is included into $\rcd$ by the recursive rules. The framework is to extract $O_1$ and $(\tod_1,k_1)$ (and $O_2$, $(\tod_2,k_2)$ for join nodes) satisfying the expected properties for the child node(s) and shows that $(\tod,k)$ will be add into $\rcd$ because of $(\tod_1,k_1)$ (and $(\tod_2,k_2)$). By induction, the extracted record will be included by the algorithm because they satisfy the expected properties, so $(\tod,k)$ will also be included. For the 'only if' part we have $(\tod,k)$ included and aim to prove the existence of a satisfying $O$. The framework is to take the record(s) based on which $(\tod,k)$ is added. By induction, the record(s) we take has corresponding orientation(s) that satisfies the expected properties. We build $O$ according to this(these) orientation(s).

\section{Proof of Theorem~\ref{correctness}} \label{mainproof1}

Before the main proof, we prove Lemma~\ref{monotone} and Lemma~\ref{error}. Remember that Lemma~\ref{monotone} states that if $(\tod,k)\in \rcd$ then $(\tod',k')\in \rcd$ for all $(\tod',k')$ with $d(v)\geq d'(v),\forall v\in X_\alpha$ and $k\leq k'\leq |Y_\alpha|$; Lemma~\ref{error} states that $[a'+b']_{\epsilon}\sim_{\epsilon_{h+1}} (a+b)$ for $a,b,a',b'\in \mathbb{R}$ satisfying $a'\sim_{\epsilon_{h}} a$, $b'\sim_{\epsilon_h} b$ for $\epsilon_h\in [0,1]$.  

\begin{proof} \textbf{(Lemma~\ref{monotone})}
    Let $O$ be the orientation for $(d,k)$. For each $v$, we arbitrarily select $d(v)-d'(v)$ out neighbors of $v$ and reverse each edge between one selected neighbor and $v$. Let the obtained orientation be $O_1$. We show that $O_1$ and $(d',k')$ satisfies the properties. (1) and (3) are trivial. To see (2), observe that $D^-_{O_1}(v)\leq D^-_{O}(v)$ for all $v\in Y_\alpha$.
\end{proof}

\begin{proof} \textbf{(Lemma~\ref{error})}
    $a'+b'\in [a/(1+\epsilon_h)+b/(1+\epsilon_h),a(1+\epsilon_h)+b(1+\epsilon_h)]$, that is, $(a'+b')\sim_{\epsilon_h} (a+b)$. As $[a'+b']_\epsilon\sim_{\epsilon} (a'+b')$, we have $\max\{[a'+b']_\epsilon/(a+b),(a+b)/[a'+b']_\epsilon\}\leq (1+\epsilon)(1+\epsilon_h)=1+\epsilon_{h+1}+\epsilon_h \epsilon -\epsilon\leq 1+\epsilon_{h+1}$. Thus
    $[a'+b']_\epsilon\sim_{\epsilon_{h+1}}(a+b)$.
\end{proof}

In the following we start the main proof.
Leaf nodes satisfy property (A) and (B) since $\rcd=\rcdap$ for a leaf node $\alpha$. Fix a node $\alpha$ of height $h$, by induction, we assume that every node descendent to $\alpha$ satisfies (A) and (B). Now we prove $\alpha$ satisfies both (A) and (B).

\subsection{Proof of (A)}

Recall that we have some $(\tod,k)\in \rcd$ now and we aim to show the existence of some $(\hat{\tod},\hat{k})\in \rcdap$ which is $h$-close to $(\tod,k)$. The case for leaf node is trivial. There are three other cases:

\begin{description}
    
\item[{Introducing $v$ Node}.]
Suppose $\alpha$ is an   introducing $v$  node and $\alpha_1$ is its child, then we have a certificate $(\tod_1,k_1)\in \rcds{1}$, where $\tod_1=\tod\setminus v$, $k_1=k$. By the induction hypothesis, there exists a record $ (\hat{\tod}_1,\hat{k}_1)\in \rcdaps{1}$ which is $(h-1)$-close to $(\tod_1,k_1)$. 
By the recursive algorithm for $\hat{R}$, there exists $ (\hat{\tod},\hat{k})\in \rcdap$, where $\hat{\tod}\setminus v=\hat{\tod}_1,\hat{\tod}(v)=0$ and $\hat{k}=\hat{k}_1$. 
Note that for all $ u\in X_{\alpha}\setminus \{v\},\hat{\tod}(u)=\hat{\tod}_1(u)\sim_{\epsilon_{h-1}}\tod_1(u)=\tod(u)$, thus we have $\hat{\tod}(u)\sim_{\epsilon_h}\tod(u)$. And $\hat{\tod}(v)=0=\tod(v)$. Since $\hat{k}=\hat{k}_1\sim_{\delta_{h-1}}k_1=k$, we get $k\sim_{\delta_h} \hat{k}$. So $(\hat{\tod},\hat{k})$ is $h$-close to $(\tod,k)$.

\item[{Join Node}.] 
If $\alpha$ is a join node with children $\alpha_1$ and $\alpha_2$, then we have a certificate $(\tod_1,k_1)\in \rcds{1}$ and $(\tod_2,k_2)\in \rcds{2}$, where for all $v\in X_{\alpha}, \tod_1(v)+\tod_2(v)=\tod(v)$ and $k_1+k_2=k$. By the induction hypothesis, there exist $(\hat{\tod}_1,\hat{k}_1)\in \rcdaps{1}$ and $(\hat{\tod}_2,\hat{k}_2)\in \rcdaps{2}$ which are $(h-1)$-close to $(\tod_1,k_1)$ and $(\tod_2,k_2)$ respectively. Note that
$(\hat{\tod}_1,\hat{k}_1),(\hat{\tod}_2,\hat{k}_2)$ is a valid certificate, so there exists $(\hat{\tod},\hat{k})\in \rcdap$, where for all $v\in X_\alpha$, $\hat{\tod}(v)=[\hat{\tod}_1(v)+\hat{\tod}_2(v)]_\epsilon$ and $\hat{k}=\hat{k}_1+\hat{k}_2$. By Lemma \ref{error}, for all $v\in X_\alpha$, $\hat{\tod}(v)\sim_{\epsilon_h} \tod(v)$ and $\hat{k}\sim_{\delta_h} k$.

\item[{Forgetting Node}.]
If $\alpha$ is a forgetting $v$ node with child $\alpha_1$, then we have a certificate $(\tod_1,k_1)\in \rcds{1}$ which satisfies one of the following conditions: 
\begin{description}
    \item[(1)] $\tod_1(v)=|N(v)\cap Y_{\alpha}|, \tod_1\setminus v = \tod$ and $k_1=k$.
    \item[(2)] There exist $ \Delta(v)\subseteq N(v)\cap X_\alpha$ and $A\in [|N(v)\cap Y_\alpha|-c(v)+|\Delta(v)|,|N(v)\cap Y_\alpha|]$ such that for all $u\in \Delta(v),\tod_1(u)=\tod(u)-1$ and for all $ u\in X_{\alpha_1}\setminus (\Delta(v)\cup \{v\}), \tod_1(u)=\tod(u), \tod_1(v)=A$ and $k_1= k-1$.
\end{description}
 Notice that these two conditions just correspond to the recursive rules with the same index.
By the induction hypothesis, there exists an approximate counterpart of the certificate. To be specific, there exists $(\hat{\tod}_1,\hat{k}_1)\in \rcdaps{1}$ which is $(h-1)$-close to $(\tod_1,k_1)$. Consider two sub-cases:
\begin{description}
    \item \textbf{Type (1) certificate.} As $(\hat{\tod}_1,\hat{k}_1)$ is $(h-1)$-close to $(\tod_1,k_1)$ and $d_1(v)=|N(v)\cap Y_\alpha|$, we have $\hat{\tod}_1(v)\sim_{\epsilon_{h-1}} |N(v)\cap Y_\alpha|$, which means (1a) is satisfied. Let $(\tod_t,|Y_{\alpha_1}|)$ be the tested pair in (1b). By the definition of $(\tod_t,|Y_{\alpha_1}|)$, for all $u\in X_{\alpha_1}\setminus \{v\}, \tod_t(u)=\lceil \hat{\tod}_1(u)/(1+\epsilon_{h-1}) \rceil \leq \lceil(1+\epsilon_{h-1})\tod_1(u)/(1+\epsilon_{h-1})\rceil = \tod_1(u)$, and $\tod_t(v)=\tod_1(v)=|N(v)\cap Y_\alpha|$.
    Also observe that $k_1\leq |Y_{\alpha_1}|$. Thus by Lemma \ref{monotone}, $(\tod_t,|Y_{\alpha_1}|)\in \rcds{1}$, which means (1b) is satisfied. As (1a), (1b) are satisfied, there exists $ (\hat{\tod},\hat{k})\in \rcdap$, where $\hat{\tod}=\hat{\tod}_1\setminus v,\hat{k}=\hat{k}_1$. Finally, observe that for all $ u\in X_\alpha,\tod(u)=\tod_1(u)\sim_{\epsilon_{h-1}}\hat{\tod}_1(u)=\hat{\tod}(u)$. $k=k_1\sim_{\delta_{h-1}}\hat{k}_1=\hat{k}$. So $(\tod,k)$ and $(\hat{\tod},\hat{k})$ are $h$-close.
    
    \item \textbf{Type (2) certificate.} As $(\hat{\tod}_1,\hat{k}_1)$ is $(h-1)$-close to $(\tod_1,k_1)$ and $d_1(v)=A$, we have $\hat{\tod}_1(v) \geq A/(1+\epsilon_{h-1})$, which means (2a) is satisfied. Let $(\tod_t,|Y_{\alpha_1}|)$ be the tested pair in (2b), i.e. for all $u\in X_{\alpha_1}\setminus \{v\},\tod_t(u)=\lceil \hat{\tod}_1(u)/(1+\epsilon_{h-1}) \rceil$ and $\tod_t(v) = A$. Similarly we have that $\tod_1(u)\geq \tod_t(u)$ for all $u\in X_{\alpha}$ while $k_1\leq |Y_{\alpha_1}|$. Thus by Lemma \ref{monotone}, $(\tod_t,|Y_{\alpha_1}|)\in \rcds{1}$, which means (2b) is satisfied. As (2a), (2b) are satisfied, there exists $(\hat{\tod},\hat{k})\in \rcdap$, where $\hat{\tod}(u)=[\hat{\tod}_1(u)+1]_\epsilon$ for all $u\in X_\alpha\setminus\Delta(v)$,
    $\hat{\tod}(u)=\hat{\tod}_1(u)$ for all $u\in \Delta(v)$, and $\hat{k}=\hat{k}_1+1$. 
    For each $u\in \Delta(v)$, $\tod(u)=\tod_1(u)\sim_{\epsilon_{h-1}}\hat{\tod}_1(u)=\hat{\tod}(u)$; for all $u\in X_\alpha\setminus \Delta(v)$, $\tod(u)\sim_{\epsilon_h}\hat{\tod}(u)$ by Lemma \ref{error}; $k-1=k_1\sim_{\delta_{h-1}}\hat{k}_1=\hat{k}-1$ and thus $k\sim_{\delta_h} \hat{k}$. So $(\tod,k)$ and $(\hat{\tod},\hat{k})$ are $h$-close.
\end{description}
\end{description}

\subsection{Proof of (B)}

Now we have some $(\hat{\tod},\hat{k})\in \rcdap$ and we aim to show the existence of some $(\tod,k)\in \rcd$ which is $h$-close to $(\hat{\tod},\hat{k})$.

\begin{description}
\item[{Introducing $v$ Node}.]
Suppose $\alpha$ is an introducing $v$ node with $\alpha_1$ as its child, then by the the recursive rules we have a certificate $(\hat{\tod}_1,\hat{k}_1)\in \rcdaps{1}$, where $\hat{\tod}_1=\hat{\tod}\setminus v$, $\hat{k}_1=\hat{k}$. By induction hypothesis, there exists $ (\tod_1,k_1)\in \rcds{1}$ which is $(h-1)$-close to $(\hat{\tod}_1,\hat{k}_1)$. 
$(\tod_1,k_1)$ is a valid certificate, so there exists $(\tod,k)\in \rcd$, where $\tod\setminus v=\tod_1,\tod(v)=0$ and $k=k_1$. For all $u\in X_{\alpha}\setminus \{v\},\tod(u)=\tod_1(u)\sim_{\epsilon_{h-1}}\hat{\tod}_1(u)=\hat{\tod}(u)$ so $\hat{\tod}(u)\sim_{\epsilon_h}\tod(u)$; $\hat{\tod}(v)=0=\tod(v)$; $k=k_1\sim_{\delta_{h-1}}\hat{k}_1=\hat{k}$, so $k\sim_{\delta_h} \hat{k}$.

\item[{Join Node}.]
If $\alpha$ is a join node with $\alpha_1$ and $\alpha_2$ as its children, then we have a certificate $(\hat{\tod}_1,\hat{k}_1)\in \rcdap,(\hat{\tod}_2,\hat{k}_2)\in \rcdaps{2}$, where for all $v\in X_{\alpha}, [\hat{\tod}_1(v)+\hat{\tod}_2(v)]_\epsilon=\hat{\tod}(v)$ and $\hat{k}_1+\hat{k}_2=\hat{k}$.  By induction hypothesis,  there exist $ (\tod_1,k_1)\in \rcds{1},(\tod_2,k_2)\in \rcds{2}$ which are $(h-1)$-close to $(\hat{\tod}_1,\hat{k}_1)$ and $(\hat{\tod}_2,\hat{k}_2)$ respectively.
Since $(\tod_1,k_1),(\tod_2,k_2)$ is a valid certificate, we have there exists $(\tod,k)\in \rcd$, where for all $v\in X_\alpha, \tod(v)=\tod_1(v)+\tod_2(v)$ and $k=k_1+k_2$. By Lemma \ref{error}, for all $v\in X_\alpha, \tod(v)\sim_{\epsilon_h} \hat{\tod}(v)$. And $k\sim_{\delta_h} \hat{k}$.

\item[{Forgetting $v$ Node}.] If $\alpha$ is a forgetting $v$ node, then we have a certificate $(\hat{\tod}_1,\hat{k}_1)\in \rcdaps{1}$ and a tested pair $(\tod_t,|Y_{\alpha_1}|)\in \rcds{1}$ in (1b) or (2b) with one of the following types:
    \begin{description}
        \item[(1)] $\hat{\tod}_1(v)\sim_{\epsilon_{h-1}} |N(v)\cap Y_{\alpha}|; \hat{\tod}_1\setminus v = \hat{\tod}$; $\hat{k}_1=\hat{k}$; $\tod_t(v)=|N(v)\cap Y_{\alpha}|$;
        \item[(2)] there exists $\Delta(v)\subseteq N(v)\cap X_\alpha$ and $A\in [|N(v)\cap Y_\alpha|-c(v)+|\Delta(v)|,|N(v)\cap Y_\alpha|]$ such that for all $u\in \Delta(v),\hat{\tod}(u)=[\hat{\tod}_1(u)+1]_\epsilon$; for all $u\in X_{\alpha_1}\setminus \Delta(v)\cup \{v\}, \hat{\tod}_1(u)=\hat{\tod}(u); \hat{\tod}_1(v)\geq A/(1+\epsilon_{h-1})$; $\hat{k}_1= \hat{k}-1$; $\tod_t(v)=A$.
    \end{description}
    In both types, for all $u\in X_{\alpha_1}\setminus\{v\},\tod_t(u) = \lceil\hat{\tod}_1(u)/(1+\epsilon_{h-1})\rceil$. Notice that these two types just correspond to the recursive rules with the same index. 
    By induction hypothesis, there exists $ (\tod_1,k_1)\in \rcds{1}$ which is $(h-1)$-close to $(\hat{\tod}_1,\hat{k}_1)$.
By the definition of $(h-1)$-closeness we have that for every $ u\in X_{\alpha_1}\setminus \{v\}, \tod_1(u)\geq \lceil \hat{\tod}_1(u)/(1+\epsilon_{h-1}) \rceil =\tod_t(u)$. Consider the two cases:
\begin{description}
    \item[{Type (1) certificate and tested pair.}] In this case $\tod_t(v)=|N(v)\cap Y_\alpha|$ and $\hat{\tod}_1(v)\sim_{\epsilon_{h-1}} |N(v)\cap Y_\alpha|$. Notice that for all $u\in X_{\alpha_1}\setminus \{v\},\tod_t(u)=\lceil \hat{\tod}_1(u)/(1+\epsilon_{h-1}) \rceil\leq \tod_1(u)$. Consider the pair $(\tod_t,k_1^*)$ where $k_1^*=k_1+|N(v)\cap Y_\alpha|-\tod_1(v)$. As $(\tod_1,k_1),(\tod_t,|Y_{\alpha_1}|)\in \rcds{1}$, by Lemma \ref{monotone} and \ref{extent}, we have $(\tod_t,k_1^*)\in \rcds{1}$. This is a valid certificate as $\tod_t(v)=|N(v)\cap Y_\alpha|$. So there exists $(\tod,k)\in \rcd$, where $\tod=\tod_t\setminus v$ and $k=k_1^*$. 
    
    Then we show that $(\tod,k)$ is $h$-close to $(\hat{\tod},\hat{k})$. Notice that $\hat{k}=\hat{k}_1\sim_{\delta_{h-1}} k_1,k=k_1^*=k_1+|N(v)\cap Y_\alpha|-\tod_1(v)$. As $\tod_1(v)\sim_{\epsilon_{h-1}}\hat{\tod}_1(v)$, thus $\tod_1(v)\geq |N(v)\cap Y_\alpha|/(1+\epsilon_{h-1})^2$, thus we have that $|N(v)\cap Y_\alpha|-\tod_1(v)\leq ((1+\epsilon_{h-1})^2-1)\tod_1(v)\leq 3\epsilon_{h-1}k_1$. Notice that $\tod_1(v)\leq k_1$ by Lemma \ref{lb}. So $k\sim_{3\epsilon_{h-1}} k_1\sim_{\delta_{h-1}} \hat{k}_1=\hat{k}$. 
    As $(1+3\epsilon_{h-1})(1+\delta_{h-1})=1+(4h+6)(h-1)\epsilon+24h(h-1)^2\epsilon^2\leq 1+4h(h+1)\epsilon$, we have $\hat{k}\sim_{\delta_h} k$. 

    For all $u\in X_{\alpha}$, we just have $\tod(u)=\tod_t(u)\sim_{\epsilon_{h-1}} \hat{\tod}_1(u)=\hat{\tod}(u)$.
        
    \item[{Type (2) certificate and tested pair.}]In this case, there exists $\Delta(v)\subseteq N(v)\cap X_\alpha$ and $A\in [|N(v)\cap Y_\alpha|-c(v)+|\Delta(v)|,|N(v)\cap Y_\alpha|]$ such that $\tod_t(v)= A$. Still we have that for all $u\in X_{\alpha_1}\setminus \{v\}, \tod_t(u)\leq \tod_1(u)$. Let $k_1^*:= k_1+\max\{0,A-\tod_1(v)\}$. 
    As $(\tod_1,k_1),(\tod_t,|Y_{\alpha_1}|)\in \rcds{1}$, by Lemma \ref{monotone} and \ref{extent}, we have $(\tod_t,k_1^*)\in \rcds{1}$.
    This is a valid certificate as $\tod_t(v)= A$. So there exists $(\tod,k)\in \rcd$, where for all $u\in X_\alpha\setminus \Delta(v),\tod(u)=\tod_t(u)$, for all $u\in \Delta(v),\tod(u)=\tod_t(u)+1$ and $k=k_1^*+1$.
    
    We use the same idea to show $\hat{k}\sim_{\delta_{h}} k$.
    Still, we have $k_1\geq \tod_1(v)\geq A/(1+\epsilon_{h-1})^2$. So $k_1^*=k_1+\max\{0,A-\tod_1(v)\}\leq 3\epsilon_{h-1}k_1$ and obviously, $k_1^*\geq k_1$. So $k_1^*\sim_{3\epsilon_{h-1}}k_1$.
    As $\hat{k}-1=\hat{k}_1\sim_{\delta_{h-1}} k_1$, we have $\hat{k}-1\sim_{\delta_h} k_1^*=k-1$. Thus $\hat{k}\sim_{\delta_h} k$.

    For all $u\in X_{\alpha}\setminus \Delta(v)$, we have $\tod(u)={\tod_1}^*(u)\sim_{\epsilon_{h-1}} \hat{\tod}_1(u)=\hat{\tod}$. For all $u\in \Delta(v)$, we have $\tod(u)-1={\tod_1}^*(u)\sim_{\epsilon_{h-1}} \hat{\tod}_1(u)$ and $\hat{\tod}(u)=[\hat{\tod}_1(u)+1]_\epsilon$, by Lemma \ref{error} we have $\tod(u)\sim_{\epsilon_h} \hat{\tod}(u)$.
\end{description}
\end{description}

\noindent\textbf{Remark} The above proof actually provides the intuition of how to modify our algorithm so that it outputs a solution of size at most $(1+\delta_{h_0})^2 OPT$. The idea is to, for all $\alpha\in V(T)$ and all $(\hat{\tod},\hat{k})\in \rcdap$, keep track of an exact $h$-close record $(\tod,k)$ of $(\hat{\tod},\hat{k})$ and its corresponding orientation i.e. an orientation $O$ with which $(\tod,k)$ satisfies the expected properties. Still, this is done by a bottom-up dynamic programming. Fix a non-leaf node $\alpha$, suppose that for all its children, this has been done. Now suppose we want to find that orientation for a record $(\hat{\tod},\hat{k})\in \rcdap$. According to the recursive rules, there exists $(\hat{\tod}_1,\hat{k}_1)\in \rcdaps{1}$ (and $(\hat{\tod}_2,\hat{k}_2)\in \rcdaps{2}$ for join nodes) from which we construct $(\hat{\tod}_1,\hat{k}_1)$. Proof of (B) in fact shows that if the exact $h-1$-close exact counterpart and the corresponding orientation has been stored, then we can construct the $h$-close record $(\tod,k)\in \rcd$ and its corresponding orientation. Notice that if $\alpha$ is the forgetting node we may need Lemma~\ref{extent} to prove the existence of such $(\tod,k)$. But fortunately, Lemma~\ref{extent} is also constructive.

\label{proofapprx}

\section{Proof of Theorem~\ref{main2}}\label{proofTSSVDS}

We first prove that for any bag $X_\alpha$ in a tree decomposition for a graph $G=(V,E)$, vertex sets $Y_\alpha$ and $ V\setminus V_\alpha$ are disconnected in $G[V\setminus X_\alpha]$ i.e. $X_\alpha$ separates $V\setminus X_\alpha$ into two disconnected parts $Y_\alpha$ and $V\setminus V_\alpha$.
Assume they are connected, then there exists $u\in Y_\alpha$ and $v\in V\setminus V_\alpha$ such that $(u,v)\in E$. So there exists some bag containing both $u$ and $v$. This implies that the nodes whose assigned bags containing $u$ or $v$ forms a subtree in the tree decomposition. However, $X$ divides apart some nodes whose assigned bags containing $u$ or $v$, a contradiction.

Since $(T,\mathcal{X})$ is a tree decomposition for $G_I[V_I\setminus D]$, a corollary is that for any node $\alpha\in V(T)$, $X_\alpha\cup D$ separates $V_I\setminus (D\cup X_\alpha)$ into disconnected parts $Y_\alpha$ and $V_I\setminus(V_\alpha\cup D)$.

Now we analyze Algorithm~\ref{subprocess}. We use induction. Firstly let's consider basic cases. If $(I,D)$ has a minimum solution of size at most $l$, then the algorithm returns at line $8$ an optimal solution. If $(I,D)$ contains no solution, which is equivalent to $V_I$ is not a solution due to monotonicity, then any leaf node is not $l$-good since $Y_{\alpha'}=\emptyset$ for a leaf node $\alpha'$ and the algorithm returns at line $12$. So in these cases, the algorithm is correct. In the remaining case, the algorithm picks a node $\alpha$ which is not $l$-good at line 10, then it adds some vertices to the final output and creates a new instance to make a recursive call. Since $\alpha$ is the node which is not $l$-good node with minimum height, its children are all $l$-good. Let the optimal solution for $(I,D)$ be $S^*$. Let $S:=Solve((I,D\cup F),(T',\mathcal{X}),l)$ and let $S'$ denote the optimal solution for $(I,D\cup F)$.

\begin{lemma}\label{splitimply}
As the problem is monotone and splittable, we have the following:
\begin{enumerate}[(i)]
    \item $S^*\cap Y_{\alpha}$ is a solution for $(I,V_I\setminus Y_{\alpha})$. 
    \item For all $\alpha_c$ a child of $\alpha$, $S^*\cap Y_{\alpha_c}$ is a solution for $(I,V_I\setminus Y_{\alpha_c})$;
    \item $S^*\setminus F$ is a solution for $(I,D\cup F)$;
    \item $E'\cup S$ is a solution for $(I,D)$.
\end{enumerate}
\end{lemma}
\begin{proof}
    \begin{enumerate}[(i)]
    \item By the definition of partial instances, $S^*\cup D$ is a solution for $I$. By monotonicity, $S^*\cup D\cup (V_I\setminus Y_{\alpha})=S^*\cap Y_{\alpha}\cup (V_I\setminus Y_{\alpha})$ is also a solution for $I$. So $S^*\cap Y_{\alpha}$ is a solution for $(I,V_I\setminus Y_{\alpha})$ according to the definition of partial solution.
    \item Similarly as above, by monotonicity, $S^*\cup D\cup (V_I\setminus Y_{\alpha_c})=S^*\cap Y_{\alpha_c}\cup (V_I\setminus Y_{\alpha_c})$ is also a solution for $I$. So $S^*\cap Y_{\alpha_c}$ is a solution for $(I,V_I\setminus Y_{\alpha_c})$.
    \item By monotonicity, $S^*\cup D\cup F$ is also a solution for $I$. So $S^*\setminus F$ is a solution for $(I,D\cup F)$.
    \item We need to use the property that $\Phi$ is splittable. By the algorithm, $E'=\bigcup_{\alpha_c\in N^-_{\alpha}} X_{\alpha_c} \cup \bigcup_{\alpha_c\in N^-_{\alpha}} E_{\alpha_c}$ and $F=\bigcup_{\alpha_c\in N^-_{\alpha}} V_{\alpha_c}$. Let $X'$ denote $\cup_{\alpha_c\in N^-_{\alpha}} X_{\alpha_c}$. To use the property that $\Phi$ is splittable, observe that $D\cup X'$ is a separator. Each $Y_{\alpha_c}$ is an isolated part (not connected to the remaining graph) in $G_I[V_I\setminus (D\cup X')]$. The remaining part in $G_I[V_I\setminus (D\cup X')]$ is thus isolated and it is $V_I\setminus (D\cup X'\cup \bigcup_{\alpha_c\in N^-_{\alpha}} Y_{\alpha_c})=V_I\setminus (D\cup F)$. Because each $E_{\alpha_c}$ is a solution for $(I,V_I\setminus Y_{\alpha_c})$, and by induction hypothesis, $S$ is a solution for $(I,D\cup F)$, we get that $\Phi$ is splittable implies $X'\cup D\cup S\cup \bigcup_{\alpha_c\in N^-_{\alpha}} E_{\alpha_c}=E'\cup D\cup S$ is a solution for $I$. So $E'\cup S$ is a solution for $(I,D)$.
\end{enumerate} \end{proof}

By induction we assume that $|S|\leq (1+(w+1)/(l+1))|S'|$. 
The approximation ratio is
\begin{align*}
    \frac{|S\cup E'|}{|S^*|}\leq \frac{|S|+\sum_{\alpha_c\in N^-_\alpha}|E_{\alpha_c}|+|\bigcup_{\alpha_c\in N^-_\alpha}{X_{\alpha_c}}|}{|S^*\cap F|+|S^*\setminus F|}.
\end{align*}
Since $|S|/|S^*\setminus F|\leq |S|/|S'|\leq 1+(w+1)/(l+1)$, we only need to show $(\sum_{\alpha_c\in N^-_\alpha}|E_{\alpha_c}|+|\bigcup_{\alpha_c\in N^-_\alpha}{X_{\alpha_c}}|)/|S^*\cap F|\leq 1+(w+1)/(l+1)$. Notice that by the definition, $Y_\alpha\subseteq F$. Since $\alpha$ is not $l$-good, (i) implies that $|S^*\cap F|\geq |S^*\cap Y_\alpha|\geq l+1$.  By (ii), for all $\alpha_c\in N^-_\alpha$, $|E_{\alpha_c}|\leq |S^*\cap Y_{\alpha_c}|$. We have
\begin{align*}
    &\frac{\sum_{\alpha_c\in N^-_\alpha}|E_{\alpha_c}|+|\bigcup_{\alpha_c\in N^-_\alpha}{X_{\alpha_c}}|}{|S^*\cap F|} \\
    = &\frac{\sum_{\alpha_c\in N^-_\alpha}|E_{\alpha_c}|}{|S^*\cap F|}+ \frac{|\bigcup_{\alpha_c\in N^-_\alpha}X_{\alpha_c}|}{|S^*\cap F|}\\
    \leq &\frac{\sum_{\alpha_c\in N^-_\alpha}|E_{\alpha_c}|}{\sum_{\alpha_c\in N^-_\alpha}|S^*\cap Y_{\alpha_c}|}+ \frac{|\bigcup_{\alpha_c\in N^-_\alpha}X_{\alpha_c}|}{|S^*\cap F|}\text{ ($Y_{\alpha_c}$'s are disjoint subsets of $F$)}\\
    \leq &1+ \frac{|\bigcup_{\alpha_c\in N^-_\alpha}X_{\alpha_c}|}{|S^*\cap F|}\text{ (By (ii) and the definition of $E_{\alpha_c}$)}\\
     \leq &1+ \frac{|\bigcup_{\alpha_c\in N^-_\alpha}X_{\alpha_c}|}{l+1}\text{ (By $|S^*\cap F|\ge l+1$)}.
\end{align*}
In a nice tree decomposition, the only case that $|N^-_\alpha|>1$ is that $\alpha$ is a join node, however in this case, the bags of its two children are the same. So ${|\bigcup_{\alpha_c\in N^-_\alpha}X_{\alpha_c}|}/(l+1)+1\leq (w+1)/(l+1)+1$.
The approximation ratio follows.
Each time we make a recursive call, the optimal solution size for the current instance decreases by at least $1$. It follows that the algorithm makes at most $O(n)$ recursive calls, so the running time is $f(l,w,n)n^{O(1)}$. And thus Theorem~\ref{main2} is proved.

\begin{figure}
\begin{subfigure}[t]{0.3\textwidth}
    \centering
    \includegraphics[width=0.5\textwidth]{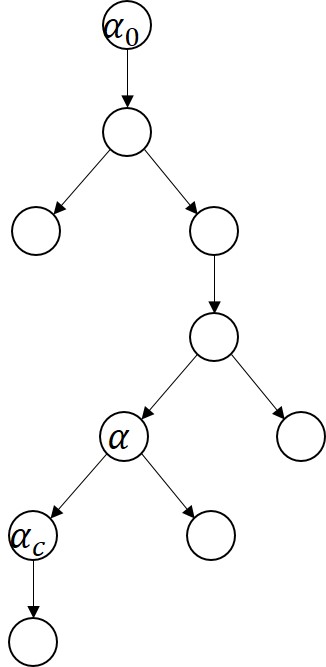}
    \caption{$T$ in a tree decomposition $(T,\mathcal{X})$.}
\end{subfigure}\hfill
\begin{subfigure}[t]{0.3\textwidth}
    \centering
    \includegraphics[width=0.7\textwidth]{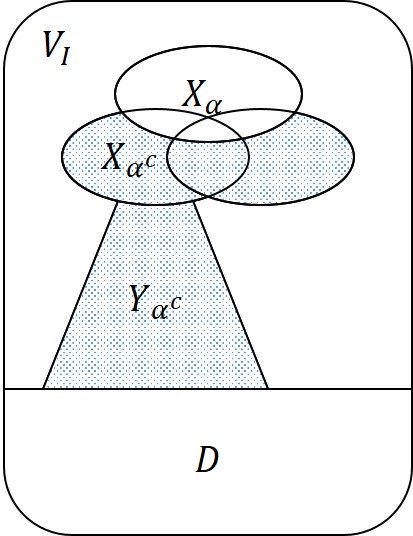}
    \caption{The vertex sets about $\alpha$ and $\alpha_c$. Dotted part is $Y_{\alpha}$}
\end{subfigure}
\begin{subfigure}[t]{0.3\textwidth}
    \centering
    \includegraphics[width=0.7\textwidth]{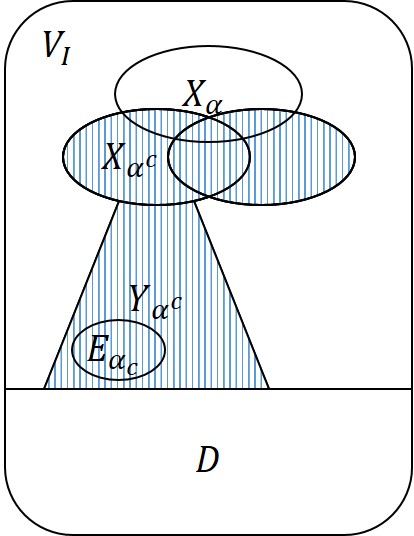}
    \caption{$E_{\alpha_c}$ is added, and $F$ is the lined part.}
\end{subfigure}
\caption{Venn diagram of sets defined in Algorithm~\ref{subprocess}}\label{fig:demAlg}
\end{figure}

\end{document}